\documentclass{llncs}  

\usepackage[utf8]{inputenc}

\usepackage{amssymb,amsmath,amsfonts}
\usepackage[all]{xy} 
\SelectTips{cm}{}
\usepackage{fancybox} 
\usepackage{color}

\usepackage{graphicx,url}

\newcommand{\pto}{\rightharpoonup} % partial 
\newcommand{\rto}{\rightarrow} 
\newcommand{\lto}{\leftarrow} 
\newcommand{\rul}[3]{#1 \lto #2 \rto #3} 

\newcommand{\Set}{\mathbf{Set}}
\newcommand{\Part}{\mathbf{Part}}
\newcommand{\A}{\mathbf{A}}
\newcommand{\G}{\mathbf{G}}
\newcommand{\AG}{\mathbf{AttG}}
\newcommand{\PAG}{\mathbf{PAttG}}
\newcommand{\Gr}{\mathbf{Gr}}
\newcommand{\C}{\mathbf{C}}

\newcommand{\Alg}{\mathbf{Alg}}
\newcommand{\Mod}{\mathbf{Mod}}
\newcommand{\Func}{\mathbf{Func}}
\newcommand{\comma}[2]{(#1\downarrow#2)} 
\newcommand{\step}[3]{\xymatrix{#1 \ar@{=>}[r]^{\mathrm{#2}} & #3 \\ } }
\newcommand{\stepp}[4]{\xymatrix{#1 \ar@{=>}[r]^{\mathrm{#2}}_{#3} & #4 \\ } }
\newcommand{\att}{\alpha} 
\newcommand{\Sp}{\mathit{Sp}} % specification 
\newcommand{\D}{\mathcal{D}} % domain of definition
\newcommand{\id}{\mathit{id}} 
 
\newcommand{\po}{\mathit{po}} 
 
\newcommand{\fpbc}{\mathit{fpbc}} 
\newcommand{\gr}{\mathit{gr}} 
\newcommand{\ag}[1]{\widehat{#1}} 
\newcommand{\ol}{\overline}
\newcommand{\ti}{\widetilde} 

\begin{document}

  \title{Transformation of Attributed Structures \\ 
         with Cloning\thanks{This work has been partly 
              funded by the project CLIMT of the French 
              \textit{Agence Nationale de la Recherche} 
              (ANR-11-BS02-016) and by the project TGV of 
              CNRS-INRIA-FAP's ($\#$ 156779).}
              (Long Version)} 
  \author{D. Duval\inst{1} and R. Echahed\inst{2} and F. Prost\inst{2}
    and L. Ribeiro\inst{3}}

 \institute{LJK - Université de Grenoble 
   \and LIG - Université de Grenoble 
   \and INF - Universidade Federal do Rio Grande do Sul}

  \maketitle \begin{abstract} 

Copying, or cloning, is a basic operation used in the specification of
many applications in computer science. However, when dealing with
complex structures, like graphs, cloning is not a straightforward
operation since a copy of a single vertex may involve (implicitly)
copying many edges. Therefore, most graph transformation approaches
forbid the possibility of cloning. We tackle this problem by
providing a framework for graph transformations with cloning. 
We use attributed graphs and allow rules to change
attributes.  These two features (cloning/changing attributes) together
give rise to a powerful formal specification approach. In order to
handle different kinds of graphs and attributes, we first define the
notion of attributed structures in an abstract way. Then we generalise
the sesqui-pushout approach of graph transformation in the proposed
general framework and give appropriate conditions under which
attributed structures can be transformed. Finally, we instantiate our
general framework with different examples, showing that many
structures can be handled and that the proposed framework allows one
to specify complex operations in a natural way.

\end{abstract}

%-------------------------------------------------------------------------------
%------------------------------------------------------------------------------
\section{Introduction}

Graph structures and graph transformation have been successfully used
as foundational concepts of modelling languages in a wide range of
areas related to software engineering.  Such a success mainly stems
from the intuitive and pictorial features of graphs which ease the
writing as well as the understanding of specifications.  Several ways
to define graph transformation rules have been proposed (see e.g.,
\cite{handbook1,handbook2,handbook3} for a survey). We can distinguish
two main approaches: The algorithmic approach which is rather
pragmatic and defines graph transformation rules by means of the
algorithms used to transform the graphs (e.g.\cite{BVG87}) and the
algebraic approach which is more abstract
(e.g. \cite{EhrigPS73}). This latter borrows notions from category
theory to define graph transformation rules.
The most popular algebraic approaches are the double pushout (DPO)
\cite{EhrigPS73,CorradiniMREHL97} and the single pushout (SPO)
%\cite{Rao84,Kennaway87,Lowe93,EhrigHKLRWC97}.
\cite{EhrigHKLRWC97}.

Very often, graph structures are endowed with attributes. Such
attributes, which enrich nodes and edges with data values, have been
proven very useful to enhance the expressiveness of visual modelling
frameworks (see, e.g., UML diagrams).  These attributes can be simple
names of an alphabet (labels) or elaborated expressions of a given
language.
Several investigations tackling attributed graph transformations have
been proposed in the literature, see
%e.g. \cite{LKW93,HabelP02,HeckelKT02,Berthold2002,EhrigEPT06,Orejas11,HabelP12,Golas12}.
e.g. \cite{LKW93,HeckelKT02,Berthold2002,EhrigEPT06,OrejasL10,Golas12}.
These proposals follow the so-called double pushout approach to define
graph transformation steps. This approach can be used in many
applications (see e.g. \cite{CorradiniMREHL97}) but it forbids actions
which consist in cloning nodes together with their incident edges
(merging of nodes is also usually forbidden).  Moreover, this approach
also prevents the application of rules that erase a node when there
are edges connected to this node in the graph that represents the
state (erasing nodes is only possible if all connected arcs are
explicitly deleted by the rule). However, there are applications in
which these restrictions of DPO would lead to rather complex
specifications.  For instance, duplicating or erasing some component
%within an architecture 
may be very useful in the development process of an architecture, and
should be a simple operation.  Also, making a security copy of a
virtual machine in a cloud (for fault-tolerance reasons) is a very
reasonable operation, as well as switching down a (physical) machine
from the infrastructure of a cloud.  To model such situations we may
profit from cloning/merging as basic operations in a formalism. But we
certainly need to use attributed structures to get a suitable
formalism for real applications.  In this paper, we propose a
framework that has both the ability to model cloning/merging of
entities in a natural way, and also the feature of using attributes
 together with the graphs.
%
%To our
%knowledge, such actions have not been investigated yet in presence of
%attributed graphs. Our aim here is to fill this gap.

To develop our proposal, we  follow a more recent approach of graph
transformation known as the sesqui-pushout approach (SqPO)
\cite{CorradiniHHK06}.  
%
%=====================================================================
%
%This latter is a conservative extension of DPO
%with some additional features such as deletion or cloning.
%
%=====================================================================
%
This latter is a conservative extension of DPO 
%in quasi-adhesive categories for left-regular rules, 
with some additional features such as deletion or cloning.
% that is when the PO complement 
%% is unique. 

A rule is defined, as in the DPO approach, by means of a span of the
form $(l:\rul{L}{K}{R}:r)$ where the morphisms $l$ and $r$ are not
necessarily monos. The fact that $l$ is not mono allows one to
duplicate some nodes and edges. Notice that most proposals dealing
with attributed graphs assume $l$ to be mono. A rewrite step can be
depicted as follows where the left square is a final pullback
complement and the right square is a pushout. The intuition is
analogous to the DPO approach: the left square specifies what is
removed (and also what is cloned) by the rule application and the
right square creates the new items. The difference, besides allowing
non injective rules, is that when applying the rule the so called
dangling condition does not need to be checked: if there are edges in
$G$ connected to nodes in the image of $m$ that is deleted by the
rule, these edges are automatically removed by the rule
application. In DPO, in such a situation a rule would not be
applicable.

\vspace{-0.5cm}

$$ \begin{array}{c}
  \xymatrix@C=4pc@R=1.5pc{ \ar@{}[rd]|{(FPBC)}  
  L \ar[d]^{m} & K \ar@{}[rd]|{(PO)} \ar[l]_{l} \ar[r]^{r} \ar[d]^{d}  
  & R \ar[d]^{h} \\
  G & D \ar[l]_{l_1} \ar[r]^{r_1} 
  & H  \\
  } \\
  \mbox{Sesqui-pushout:} \step{G}{sqpo}{H}  \\ 
\end{array} $$

In order to consider different kinds of graphs and attributes, we
present our approach in a general setting. That is to say, we consider
structures of the form $\ag{G}=(G,A,\att)$ made of an object $G$ whose
elements may be attributed, an object $A$ defining attributes and a
partial function $\att$ which assigns to some elements of $G$
attributes in $A$. The fact that $\att$ is partial turns out to be
very useful to write transformation rules that change the attributes
of some elements of $G$ (see, e.g. \cite{HabelP02,Berthold2002}).  We
do not assume $G$ to be necessarily a graph nor do we assume $A$ to be
necessarily an algebra. We thus elaborate a framework which can be
instantiated with different kinds of structures and attributes
fulfilling some criteria we introduce in this paper. Therefore we can
handle different graphs with various kinds of attributes (algebras,
lambda-terms, finite labels, syntactic theories, etc.). Similar
objectives, with different outcome, have been recently investigated in
\cite{Golas12} for the DPO approach.

The rest of the paper is organized as follows. The next section
introduces the category of attributed structures and provides some
definitions which may help the understanding of the
paper. Section~\ref{sec:pre-sqpo} recalls briefly the useful
definitions regarding the sesqui-pushout approach. Then,
Section~\ref{sec:results} shows how to lift SqPO rewriting in the
context of attributed structures. Sections~\ref{sec:lambda} and
\ref{sec:cloud} illustrate our approach through some examples while
related work are discussed in
Section~\ref{sec:relatedwork}. Concluding remarks are given in
Section~\ref{sec:conclusion}. The missing proofs may be found in
the Appendix.
%\cite{DEPR14-arxiv}.
\vspace{-0.5cm}

%-------------------------------------------------------------------------------
\section{Attributed Structures}
\label{sec:pre-att}

In this section we define the notion of attributed structures 
and set some notations.

\noindent\textbf{Structures.}
Let $\G$ be a category and $S:\G\to\Set$ a functor  
from $\G$ to the category of sets.  
For instance, $\G$ may be the category of graphs $\Gr$  \cite{handbook1}
and $S$ may be either the \emph{vertex} functor $V$ 
defined by $V(G)=V_G$ and $V(g)=g_V$, 
or the \emph{edge} functor $E$ 
defined by $E(G)=E_G$ and $E(g)=g_E$, 
or the functor $V+E$ 
which maps each graph $G$ to the disjoint union $V_G+E_G$ and 
each morphism $g:G_1\to G_2$ to the map $g_V+g_E$. 

\noindent\textbf{Attributes.} 
Let $\A$ be a category and $T:\A \to \Set$ a functor 
from $\A$ to the category of sets. 
For instance, $\A$ may be the category $\Alg(\Sigma)$ 
of $\Sigma$-algebras \cite{ST2012} 
for some signature $\Sigma=(S,\Omega)$, 
or more generally the category $\Mod(\Sp)$ of models of 
an equational specification $\Sp=(\Sigma,E)$, 
made of a signature $\Sigma$ and a set of equations $E$. 
Then the functor $T:\A \to \Set$ may be 
such that $T(A)=\sum_{s\in S}A_s$, i.e., 
$T$ maps each $\Sigma$-algebra $A$ to the disjoint union of its carriers, 
or more generally $T(A)=\sum_{s\in S'}A_s$ for some fixed subset $S'$ of $S$. 
In the following, we sometimes write $Fx$ instead of $F(x)$ 
when a functor $F$ is applied to an object or a morphism  $x$.

\begin{definition}
The category of \emph{attributed structures} $\AG$ 
(with respect to the functors $S$ and $T$) 
is the comma category $\comma{S}{T}$.
Thus, an \emph{attributed structure} is a triple 
$\ag{G}=(G,A,\att)$ made of an object $G$ in $\G$, 
an object $A$ in $\A$ and 
a map $\att: S(G) \to T(A)$ (in $\Set$) ; 
and a \emph{morphism of attributed structures} 
$\ag{g}:\ag{G}\to\ag{G'}$,
where $\ag{G}=(G,A,\att)$ and $\ag{G'}=(G',A',\att') $, 
is a pair $\ag{g}=(g,a)$ made of 
a morphism $g:G\to G'$ in $\G$ 
and a morphism $a:A\to A'$ in $\A$ 
such that $\att' \circ Sg = Ta \circ \att$ (in $\Set$).
  $$ \xymatrix@C=3pc@R=1.5pc{
  \ag{G} \ar[d]_{\ag{g}} & \ar@{}[d]|{=} & 
    G \ar[d]_{g} & SG \ar[d]_{Sg} \ar[r]^{\att} & 
    TA \ar[d]^{Ta} & A \ar[d]^{a} \\ 
  \ag{G'} & & 
    G' & SG' \ar[r]^{\att'} & TA' \ar@{}[ul]|{=} & A' \\ 
  }$$  
\end{definition}

\textbf{Partial maps.} 
Let $\Part$ be the category of sets with partial maps,
which contains $\Set$. 
A partial map $f$ from $X$ to $Y$ is denoted $f:X\pto Y$ and 
its domain of definition is denoted $\D(f)$.
The partial order between partial maps is denoted $\leq$,
it endows $\Part$ with a structure of 2-category. 
By composing $S$ and $T$ with
the inclusion of $\Set$ in $\Part$ we get two functors  
$S_p:\G\to\Part$ and $T_p:\A \to \Part$. 

\begin{definition}
\label{defi:PAG}
The category of \emph{partially attributed structures} $\PAG$
(with respect to the functors $S$ and $T$) 
is defined as follows. 
A \emph{partially attributed structure} is a triple 
$\ag{G}=(G,A,\att)$ made of an object $G$ in $\G$,
an object $A$ in $\A$ and 
a partial map $\att: S_p(G) \pto T_p(A)$ (in $\Part$) ; 
and a \emph{morphism of partially attributed structures} 
$\ag{g}:\ag{G}\to\ag{G'}$,
where $\ag{G}=(G,A,\att)$ and $\ag{G'}=(G',A',\att') $, 
is a pair $\ag{g}=(g,a)$ made of 
a morphism $g:G\to G'$ in $\G$ 
and a morphism $a:A\to A'$ in $\A$ 
such that $\att' \circ S_pg \geq T_pa \circ \att$ (in $\Part$).
% rotate for \geq? ***
  $$ \xymatrix@C=3pc@R=1.5pc{
  \ag{G} \ar[d]_{\ag{g}} & \ar@{}[d]|{=} & 
  G \ar[d]_{g} & S_pG \ar[d]_{S_pg} \ar@{-^{>}}[r]^{\att} & 
    T_pA \ar[d]^{T_pa} & A \ar[d]^{a} \\ 
  \ag{G'} & & 
  G' & S_pG' \ar@{-^{>}}[r]^{\att'} & T_pA' \ar@{}[ul]|{\geq} & A' \\ 
  }$$  
Such a morphism of partially attributed structures is called 
\emph{strict} when $\att' \circ S_p(g) = T_p(a) \circ \att$. 
\end{definition}

\begin{remark} 
Clearly, $\AG$ is a full subcategory of $\PAG$ 
and every morphism in $\AG$ is a strict morphism in $\PAG$. 
The subcategory $\AG$ of $\PAG$ is called the 
subcategory of \emph{totally attributed structures}.
\end{remark} 

\begin{definition}
A morphism of (partially) attributed structure 
$\ag{g}:\ag{G}\to\ag{G'}$ \emph{preserves attributes} % preserves ? ***
if $\ag{G}=(G,A,\att)$, $\ag{G'}=(G',A,\att')$ 
and $\ag{g}=(g,\id_A)$ for some object $A$ in $\A$.
\end{definition}

\textbf{Notations.} 
We will omit the subscript $p$ in $S_p$ and $T_p$.
%When there is no ambiguity, for each morphism $g:G\to G'$ in $\G$ 
%we may denote $g=Sg:SG\to SG'$, and similarly 
%for each morphism $a:A\to A'$ in $\A$ we may denote $a=Ta:TA\to TA'$.
% 
Let $(G,A,\att)$ be a (partially) attributed structure, 
the notation $x:t$  
means that $x \in S(G)$, $t \in T(A)$ and $\att(x)=t$ 
(i.e., $x$ has $t$ as  attribute),
and the notation $x:\bot$
means that $x \in S(G)$, $x\not\in\D(\att)$ (i.e., $x$ has no attribute). 
%\\ 
Let $(G,A,\att)$ and $(G',A',\att') $ be 
attributed structures, let $g:G\to G'$ in $\G$ and $a:A\to A'$ in $\A$, 
then $(g,a):(G,A,\att)\to(G',A',\att') $ is a morphism of 
attributed structures if and only if for all $x \in S(G)$ and $t \in T(A)$ 
$x:t \implies g(x):a(t)$. 
%\\
Let $(G,A,\att)$ and $(G',A',\att') $ be 
partially attributed structures, 
let $g:G\to G'$ in $\G$ and $a:A\to A'$ in $\A$, 
then $(g,a):(G,A,\att)\to(G',A',\att') $ is a morphism of partially
attributed structures if and only if for all $x \in SG$ and $t \in TA$ 
$x\in\D(\att) \implies g(x)\in\D(\att') \mbox{ and then }
x:t \implies g(x):a(t)$, 
and $(g,a)$ is strict if and only if for all $x \in SG$ and $t \in TA$ 
$x\in\D(\att) \iff g(x)\in\D(\att')$, and then $x:t \implies g(x):a(t)$. 
%\\
The notation $x:\bot$ can be misleading: 
of course we can extend $a:TA\to TA'$ as $a:TA+\{\bot\}\to TA'+\{\bot\}$ 
by setting $a(\bot)=\bot$, 
but then it is \emph{false} that $x:t \implies g(x):a(t)$ for each 
$x \in SG$ and $t \in TA+\{\bot\}$.
In fact, for each morphism of partially attributed structures $(g,a)$ 
we have $g(x):\bot \implies x:\bot$,
and it is only when $g$ is strict that in addition $x:\bot \implies g(x):\bot$.

\begin{definition}
\label{defi:under}
The \emph{underlying structure} functor is the functor 
$U_{\G}:\PAG\to\G$ which maps an attributed structure $(G,A,\att)$ 
to the object $G$ and $(g,a)$ to the morphism $g$. 
The \emph{underlying attributes} functor is the functor 
$U_{\A}:\PAG\to\A$ which maps an attributed structure $(G,A,\att)$ 
to the object $A$ and $(g,a)$ to the morphism $a$. 
\end{definition}

\vspace{-0.5cm}
%------------------------------------------------------------------------------
\section{Sesqui-Pushouts}
\label{sec:pre-sqpo}

% cite 2 previous papers on FPBC? ***
In this section we briefly recall the definition of sesqui-pushout (SqPO) 
rewriting, introduced in \cite{CorradiniHHK06}.
A sesqui-pushout rewriting step is made of a final pullback complement (FPBC)
followed by a pushout (PO). 
The definitions of FPBC and SqPO are reminded here,
in any category $\C$. 
The initiality property of POs and the finality property of FPBCs 
imply that POs, FPBCs and SqPOs are unique up to isomorphism,
when they exist.
% right adjoint to ``the'' pullback functor? ***

\begin{definition}
The \emph{final pullback complement} (FPBC) of 
a morphism $m_L:L\to G$ along a morphism $l:K\to L$ is a pullback (PB)
(below on the left) such that for each pullback (below on the right) 
\vspace{-0.3cm}
$$ \xymatrix@C=6pc{
L \ar[d]_{m_L} \ar@{}[rd]|{(PB)}  & 
  K \ar[l]_{l} \ar[d]^{m_K} \\
G & 
  D \ar[l]_{l_1} \\
} 
\qquad 
\xymatrix@C=6pc{
L \ar[d]_{m_L} \ar@{}[rd]|{(PB)}  & 
  K' \ar[l]_{l'} \ar[d]^{m'} \\
G & 
  D' \ar[l]_{l'_1} \\
} $$
and each morphism $f:K'\to K$ such that $l\circ f=l'$ 
there is a unique morphism $f_1:D'\to D$ such that 
$l_1\circ f_1=l'_1$ and $f_1\circ m'=m_K\circ f$. 
$$ \xymatrix@C=6pc{
L \ar[d]_{m_L} \ar@{}[rd]|{(FPBC)}   & 
  K \ar[l]^{l} \ar[d]^{m_K} & 
  K' \ar@/_3ex/[ll]_{l'} \ar[l]^{f} \ar[d]^{m'} \\ 
G & 
  D \ar[l]_{l_1} & 
  D' \ar@/^3ex/[ll]^{l'_1} \ar@{-->}[l]_{f_1} \\ 
} $$
\end{definition}

\begin{definition}
The \emph{sesqui-pushout} of 
a morphism $m_L:L\to G$ along a span of morphisms $(l:L\lto K \to R:r)$ 
is the FPBC of $m_L$ along $l$ followed by the PO of $m_K$ along $r$ 
(see diagram below).
$$   \xymatrix@C=6pc{ 
  \ar@{}[rd]|{(FPBC)}  
  L \ar[d]_{m_L} & K \ar@{}[rd]|{(PO)} \ar[l]_{l} \ar[r]^{r} \ar[d]^{m_K}  
  & R \ar[d]^{m_R}\\
  G & D \ar[l]_{l_1} \ar[r]^{r_1} 
  & H  \\
  }  $$
\end{definition}

A comparison of SqPO with DPO and SPO approaches can be found 
in \cite{CorradiniHHK06}, where it is stated that
``\emph{Probably the most original and interesting feature of sesqui-pushout rewriting is the fact that it can be applied to non-left-linear rules as well, 
and in this case it models the cloning of structures.}'' 

In the category of graphs, 
under the assumption that $m_L\colon L\to G$ is an inclusion, 
the result of the sesqui-pushout can be described as follows 
\cite[Section 4.1]{CorradiniHHK06}, \cite{DuvalEP12}.
With respect to a rule $(l:L\lto K \to R:r)$, let us 
call \emph{tri-node} a triple $(n_L,n_K,n_R)$ 
where $n_L$, $n_K$ and $n_R$ are nodes in $L$, $K$ and $R$ respectively
and where $n_L=l(n_K)$ and $n_R=r(n_K)$. 
Since $m_L$ is an inclusion, $L$ is a subgraph of $G$. 
Let $\ol{L}$ be the subgraph of $G$ made of all the nodes outside $L$
and all the vertices between these nodes. 
Let $\ti{L}$ be the set of edges outside $L$ 
with at least one endpoint in $L$ (called the \emph{linking edges}),
so that $G$ is the disjoint union of $L$, $\ol{L}$ and $\ti{L}$. 
Then, up to isomorphism, $m_R$ is an inclusion and 
$H$ is obtained from $G$ by replacing $L$ by $R$
and by ``gluing $R$ and $\ol{L}$ in $H$ according to 
the way $L$ and $\ol{L}$ are glued in $G$'',
which means precisely that $H$ is the disjoint union of $R$, $\ol{L}$ 
and the following set $\ti{R}$ of linking edges (see \cite{DuvalEP12} for more details):
\begin{itemize}
\item if $n$ is a node in $R$ and $p$ a node in $\ol{L}$,
there is an edge from $n$ to $p$ in $\ti{R}$ 
for each tri-node $(n_L,n_K,n_R)$ with $n_R=n$ 
and each edge from $n_L$ to $p$ in $\ti{L}$; 
\item if $n$ is a node in $\ol{L}$ and $p$ a node in $R$,
there is an edge from $n$ to $p$ in $\ti{R}$ 
for each tri-node $(p_L,p_K,p_R)$ with $p_R=p$ 
and each edge from $n$ to $p_L$ in $\ti{L}$; 
\item if $n$ and $p$ are nodes in $R$,
there is an edge from $n$ to $p$ in $\ti{R}$ 
for each tri-node $(n_L,n_K,n_R)$ with $n_R=n$, 
each tri-node $(p_L,p_K,p_R)$ with $p_R=p$ 
and each edge from $n_L$ to $p_L$ in $\ti{L}$. 
\end{itemize}

%-------------------------------------------------------------------------------
\vspace{-0.3cm}

\section{Attributed Sesqui-Pushout Rewriting}
\label{sec:results}

In this section 
we define rewriting of attributed structures based on sesqui-pushouts,
then we construct such SqPOs from SqPOs of the underlying (non-attributed) 
structures.

\begin{definition}
\label{def:rule}
Given an object $A$ of $\A$, a \emph{rewriting rule} 
with attributes in $A$ is a span 
$(\ag{l}:\rul{\ag{L}}{\ag{K}}{\ag{R}}:\ag{r})$,
or simply $(\ag{l},\ag{r})$, 
made of morphisms $\ag{l}$ and $\ag{r}$ in $\PAG$ which preserve attributes
and such that $\ag{L}$ and $\ag{R}$ are totally attributed structures. 
A \emph{match} for a rule $(\ag{l},\ag{r})$ in an attributed structure $\ag{G}$ 
is a morphism $\ag{m}=(m,a):\ag{L}\to\ag{G}$ in $\AG$
such that the map $Sm$ is injective. 
The \emph{SqPO rewriting step} (or simply the \emph{rewriting step})  
applying a rule $(\ag{l},\ag{r})$ 
to a match $\ag{m}$ is the sesqui-pushout of $\ag{m}$
along $(\ag{l},\ag{r})$ in the category $\PAG$. 
%\leila{Do we need to define the match as a strict morphism? 
% Doesn't this follows from the fact that $\ag{L}$ is totally attributed?}
% \dd{you are right}
\end{definition}

From the definition above, a rewrite rule is characterised by (i) the object $A$ of attributes, (ii) the attributed structures $\ag{L}$, $\ag{K}$and $\ag{R}$ and (iii) the span of structures $(l:L\leftarrow K \rightarrow R:r)$. 
A match $\ag{m}$ must have an injective 
underlying morphism of structures but it may modify the attributes.
In contrast, the morphisms $\ag{l}=(l,\id_A)$ and $\ag{r}=(r,\id_A)$ 
in a rule have arbitrary underlying morphisms of structures $l$ and $r$,
thus allowing items to be added, deleted, merged or cloned,  
but they must preserve attributes
since their underlying morphism on attributes is the identity $\id_A$.
However, since $\ag{K}$ is only partially attributed, 
any element $x\in SK$ without attribute may be mapped 
to $l(x):a$ in $\ag{L}$ and to $r(x):a'$ in $\ag{R}$ with $a\ne a'$. 
Thus the assignment of attributes to vertices/edges may
change in the transformation process.

In the following when $(m,a)$ is a match we often assume that 
$Sm$ is an inclusion, rather than any injection; in this way 
the notations are simpler 
while the results are the same, since all constructions 
(PO, PB, FPBC) are up to isomorphism. 

The construction of a sesqui-pushout in $\PAG$ can be made in two steps:
first a sesqui-pushout in $\G$, 
which depends only on the properties of the category $\G$,
then its lifting to $\PAG$, which does not depend any more on $\G$. 
Moreover, this lifting is quite simple: 
since the morphisms $l$ and $r$ do not modify the attributes, 
it can be proved that $m_K$ and $m_R$ have the same 
underlying morphism on attributes as $m_L$.
This is stated in Theorem~\ref{theo:sqpo}. 

\vspace{-0.2cm}
\begin{theorem}
\label{theo:sqpo}
Let us assume that the functors 
$U_{\G}:\PAG\to\G$, $U_{\A}:\PAG\to\A$, $S:\G\to\Set$ and $T:\A\to\Set$ 
preserve PBs and that the functor $S$ preserves POs. 
Let $(\ag{l}:\rul{\ag{L}}{\ag{K}}{\ag{R}}:\ag{r})$ 
be a rewriting rule 
and $\ag{m_L}=(m_L,a):\ag{L}\to \ag{G}$ a match. 
If diagram $\Delta$ (below on the left) 
is a SqPO rewriting step in $\G$ 
then diagram $\ag{\Delta}$ (below on the right) 
is a SqPO rewriting step in $\PAG$ and $(m_R,a)$ is a match.
\end{theorem}
\vspace{-0.3cm}
$$ \begin{array}{lll}
\Delta: & \quad & \ag{\Delta}: \\ 
\xymatrix@C=4pc{ \ar@{}[rd]|{(FPBC)}  
  L \ar[d]^(.3){m_L} & 
  K \ar@{}[rd]|{(PO)} \ar[l]_{l} \ar[r]^{r} \ar[d]^(.3){m_K}  
  & R \ar[d]^(.3){m_R}\\
  G & D \ar[l]_{l_1} \ar[r]^{r_1} 
  & H   \\
  } &&
\xymatrix@C=4pc{ \ar@{}[rd]|{(FPBC)}  
  \ag{L} \ar[d]^(.3){(m_L,a)} & 
    \ag{K} \ar@{}[rd]|{(PO)} 
    \ar[l]_{(l,\id_A)} \ar[r]^{(r,\id_A)} \ar[d]^(.3){(m_K,a)}  
  & \ag{R} \ar[d]^(.3){(m_R,a)} \\
  \ag{G} & 
    \ag{D} \ar[l]_{(l_1,\id_{A_1})} \ar[r]^{(r_1,\id_{A_1})} 
  & \ag{H} \\
 } \\ 
\end{array} $$

\begin{proof} 
Since a sesqui-pushout is a FPBC followed by a PO, 
this proof relies on similar 
results about the lifting of FPBCs and the lifting of POs 
in Appendix~\ref{app:po} and~\ref{app:fpbc}, respectively.
%(see \cite{DEPR14-arxiv}).
\end{proof} 

Let us summarize  what may occur 
for an element $x\in SD$.
If $x \not\in SK$ then only one case may occur: 
\vspace{-0.3cm}
$$
\xymatrix@R=1pc@C=1.5pc{
l_1(x):t_1 & x:t_1 \ar@{|->}[l] \ar@{|->}[r] & r_1(x):t_1 \\ 
}$$
If $x\in SK$ then two cases may occur:
\vspace{-0.3cm}
$$
\xymatrix@R=1pc@C=1.5pc{
l(x)\colon\! t \ar@{|->}[d] & 
  x\colon\! t \ar@{|->}[l] \ar@{|->}[r] \ar@{|->}[d] & 
  r(x)\colon\! t \ar@{|->}[d] \\ 
l_1 (x)\colon\! a(t) & x\colon\! a(t) \ar@{|->}[l] \ar@{|->}[r] & r_1(x)\colon\! a(t) \\ 
}
\quad\;\;
\xymatrix@R=1pc@C=1.5pc{
l()x\colon\! t \ar@{|->}[d] & 
  x\colon\! \bot \ar@{|->}[l] \ar@{|->}[r] \ar@{|->}[d] &  
  r(x)\colon\! t' \ar@{|->}[d] \\ 
l_1(x)\colon\! a(t) & x\colon\! \bot \ar@{|->}[l] \ar@{|->}[r] & r_1(x)\colon\! a(t') \\ 
}
$$

\section{Graph Transformations with Simply Typed 
              $\lambda$-terms as Attributes} % ??***
\label{sec:lambda}

\newcommand{\variables}{X}
\newcommand{\btype}{\iota}
\newcommand{\naturalchurch}{{\mathcal N}}

In this section we consider simply typed $\lambda$-terms as
attributes. The choice of the $\lambda$-calculus can be motivated by
the possibility to perform higher-order computations (functions can be
passed as parameters). 
 We refer  to \cite{Bar13} for more details concerning the simply-typed
$\lambda$-calculus, though basic notions of
$\lambda$-calculus are enough to understand the example provided
in this section.

First, let us  choose the categories $\G$ and
$\A$ and the functors $S$ and $T$.
% \G 
Let $\G=\Gr$ be the category of graphs.
% S 
Let $S:\G\to\Set$ be the functor which maps each graph to the disjoint
union of its set of vertices and its set of edges.  We define the
category $\A$ as the category where objects are sets $\Lambda(\variables)$
of simply typed $\lambda$-terms, \`a la
Church, built over  variables in $\variables$.  For the sake of simplicity we only consider
one base type $\btype$.  Simply typed $\lambda$-terms in
$\Lambda(\variables)$, noted $t$, and types, noted $\tau$, are defined
inductively by: $\tau \colon\!\!\colon\!\!=   \btype \mid \tau \to \tau$ and 
$t \colon \!\!\colon\!\! =  x \mid (t\ t) \mid \lambda x^{\tau}.t$ with $x \in \variables$.
%
%$$\begin{array}{lcl} 
%   \tau & ::= & \btype \mid \tau \to \tau \\
%   t & ::= & x \mid (t\ t) \mid \lambda x:\tau.t  \mbox{ with }x \in \variables
%\end{array}$$
%
A morphism $m$ from $\Lambda(\variables)$ to $\Lambda(\variables')$ is
totally defined by a substitution from $\variables$ to
$\Lambda(\variables')$.
The functor $T:\A \to \Set$ is such that $T(\Lambda(\variables))$ is
the set of normal forms of elements in $\Lambda(\variables)$. Other choices for
$T$ are possible, for instance $T(\Lambda(\variables))$ could be chosen
to be $\Lambda(\variables)$ itself. However in this case there would
be no reduction in the attributes while rewriting.
With the definitions as above, the functors $U_{\G}:\PAG\to\Gr$,
$U_{\A}:\PAG\to\A$, $S:\Gr\to\Set$ and $T:\A\to\Set$ preserve
pullbacks, and $S$ preserves pushouts.

Graph transformations can be coupled with $\lambda$-term evaluation.
For instance, a vertex, $n$, of a right-hand side, $R$, of a rule may
be attributed with a lambda-term, $t$, containing free variables which
occur in the left-hand side $L$.  A match, $\sigma$ of such a rule
instantiates the free variables. Firing the rule will result in (i) the
computation of the normal form of the lambda-term $\sigma(t)$ and
(ii) its attribution to the image of vertex $n$ in the resulting transformed
graph.  Below we give an example of such a rule and illustrate it on
the graph $\lambda G$.

\hspace{-1cm}
\begin{tabular}{|c|c|c|}
\hline 
$L$ & $K$ & $R$ \\
$\xymatrix@C=1pc@R=1pc{
   \ovalbox{n:$x$} \ar[d] 
   & \ovalbox{m:$y$} \\
   \ovalbox{p:$f$} }$&
$\xymatrix@C=1pc@R=1pc{\ovalbox{n:$\bot$} & \ovalbox{m:$y$}  \\
                       \ovalbox{p:$f$} & \ovalbox{m':$\bot$}}$
&$\xymatrix@C=1pc@R=1pc{
             \ovalbox{n:$(f \ x \ y)$}
             & \ovalbox{m:$y$}  \\
             \ovalbox{p:$f$} & \ovalbox{m':$f$}}$ \\
\hline
 $\lambda G$  & $\lambda D$ & $\lambda H$ \\
$\xymatrix@C=1pc@R=1pc{
   \ovalbox{n:$w$ \ar[d]} 
   & \ovalbox{m:$\lambda u^{\iota}.u$} \ar[l]\\
   \ovalbox{p:$\lambda s^{\iota} . \lambda t^{\iota}.s$} \ar[ur]}$&
$\xymatrix@C=1pc@R=1pc{\ovalbox{n:$\bot$} & 
                       \ovalbox{m:$\lambda u^{\iota}.u$} \ar[l]  \\
                       \ovalbox{p:$\lambda s^{\iota} .
                                 \lambda t^{\iota}.s$} \ar[ur] \ar[r]& 
                       \ovalbox{m':$\bot$} \ar[ul]}$
&$\xymatrix@C=1pc@R=1pc{
             \ovalbox{n:$w$}
             & \ovalbox{m:$ \lambda u^{\iota}.u$} \ar[l] \\
             \ovalbox{p:$\lambda s^{\iota} 
                                 .\lambda t^{\iota}.s$}\ar[ur] \ar[r] & 
             \ovalbox{m':$\lambda s^{\iota} 
                                 .\lambda t^{\iota}.s$\ar[ul]} & 
             }$ \\
\hline
\end{tabular}

Graph morphisms are represented via vertex name sharing, and $U_{\A}$
can be deduced from them (for instance attribute $x$ in $L$ is
instantiated by attribute $w$ in $\lambda G$ because of the match on vertex $n$,
likewise $f$ is instantiated by $\lambda s^{\iota}. \lambda
t^{\iota}.s$ and $y$ is instantiated by $\lambda u^{\iota}.u$).  In
this example several features of our framework are underlined. First,
notice that vertex $m$ in $L$ is cloned, as a structure, into $m$ and
$m'$.  This cloning of structure implies that the edges incident to
$m$ in $\lambda G$ are to be duplicated for $m$ and $m'$ in $\lambda H$.  As
for attributes, the example shows that the structure can be cloned while
the attributes can be changed (this is the case for the attribute of
vertex $m'$).
The edge between
vertices $n$ and $p$ is erased since it is matched and is not present in
$K$ nor in $R$. Furthermore, the attribute of $n$ in $R$ shows a higher-order
computation. Via the match, $f$ is substituted by the function $\lambda
s^{\iota}. \lambda t^{\iota}.s$ and is applied to the instances of $x$ and $y$. 
%Notice that $f$ is the
%first projection and correspond to the encoding of true in
%$\lambda$-calculus. The attribute of $n$ in $R$ can be seen as an
%encoding of the if-then-else structure: if $f$ is substituted by the
%encoding of true (resp. false) it will normalise to $x$ (resp. $y$).
In $\lambda H$ the attribute of $n$ is the normal form of $(\lambda
s^{\iota}.\lambda t^{\iota}.s \ w \ \lambda u^{\iota}.u)$ which is
$w$. 
Attributes can be easily copied, e.g., $f$ occurs twice
in $R$. Finally, attributes of a vertex can be modified thanks to the
partiality of the attribution in $K$. It is witnessed on vertices $n$ and
even $m'$ which is a clone of $m$. In fact $m'$ clones only the
incident edges of $m$, one would have to write $m'\colon y$ to copy
the attribute of $m$ as well.
Free variables are used to provide arguments of lambda-terms. This allows us to
simulate the attribute dependency relation introduced in
%\cite{Boi13,BoiFerSol11}.
\cite{BoiFerSol11}.

%------------------------------------------------------------------------------
%\section{Example of attributed graph transformations}
%\label{sec:examples}
%
%\begin{center}
%\fbox{HERE: the cloud example}
%\end{center}

% !TEX root = ./DEPR_FASE14_attgraphsCloud.tex
\vspace{-0.3cm}
\section{Graph Transformations with Attributes 
Defined Equationally: Administration of Cloud Infrastructure}
\label{sec:cloud}

\newcommand{\sort}[1]{{\sf #1}}
\newcommand{\opn}[1]{{\it #1}}
\newcommand{\rulename}[1]{\textit {\textbf{#1}}}
In this section we explore how our framework allows us to take into
account attributed graph transformations with attributes built 
over equational specifications. First we instantiate the definition with appropriate
categories and functors, and then model an example.

%First, let us check that the assumptions of Theorem~\ref{theo:sqpo}
%are satisfied. The functors 
%$U_{\G}$ and $U_{\A}$ are known from Definition~\ref{defi:under},
%we have to choose the categories $\G$ and $\A$
%and the functors $S$ and $T$. 
%% \G 
%Let $\G=\Gr$ be the category of graphs.
%% S 
%Let $S:\G\to\Set$ be the functor which maps each graph 
%to the disjoint union of its set of vertices and its set of edges. 
%% (and which is defined accordingly on graph morphisms). 
%% \A
%
%Let $\Sp$ be some fixed equational specification 
%and let $\A=\Mod(\Sp)$ be the category of models of $\Sp$.
%% T
Let the category $\G$ and functor $S:\G\to\Set$ be defined as in section
\ref{sec:lambda}.
Let $T:\A\to\Set$ be the functor which maps each model of $\Sp=(\Sigma,E)$, with  $\Sigma=(S,\Omega)$,
to the disjoint union of the carriers sets $A_s$ for $s$ in 
some given set of sorts $S$.
% (and which is defined accordingly on morphisms of models of $\Sp$). 
%
%\begin{proposition}
With the definitions as above, 
the functors $U_{\G}:\PAG\to\Gr$, $U_{\A}:\PAG\to\Mod(\Sp)$, 
$S:\Gr\to\Set$ and $T:\Mod(\Sp)\to\Set$ preserve pullbacks,  
and $S$ preserves pushouts.
%\end{proposition}

%%%
%%% Exemple d'orejas : couper l'arc de poids le plus faible
%%% Orejas + Lambers 2010
%%%

\newcommand{\naturalt}{{\sf N}}
\newcommand{\booleant}{{\sf B}}
\newcommand{\succc}{{\sf Succ}}
Cloud Computing is very popular nowadays \cite{cloudcomputing}. The general idea is that there is a pool, called cloud, of resources (equipment, services, etc.) that may be requested by users. A user may, for example,  request a machine with some specific configuration and services from the cloud. The cloud administrator chooses an actual physical machine that is available and installs on it  a virtual machine (short VM) according to the user specification. The user does not have to know neither where this machine is nor how the services are implemented, communication with his machine is done via the cloud. The cloud administrator has many tasks to perform, besides communicating with the clients (users). Typical operations involve load balance among the machines, optimisation of the use of machines, etc. In the following we provide the specification using graph transformations of some operations of a cloud administrator. First we define the static structure,  defining data types and the states of the system (as attributed graphs), and then we define the operations (as rules).  Since the purpose of this case study is to show  the use of our framework, we will not describe a complete set of attributes and rules needed to specify the behaviour of a cloud administrator, but concentrate on those parts that make explicit use of the features of the approach.

\subsection{Cloud Administration: Static Part}
To model this scenario, we will use  graphs with many attributes. The approach presented in the previous sections could be easily extended to families of attributes. Alternatively, one  could use just one record attribute, but we prefer the former representation since the specification becomes more readable. The attributes that will be used are: 

\begin{description}
\item[Vertex attributes:] 
       \sort{nodeType},  represent the different entities involved in this system, that is, cloud administrator, users, machines and virtual machines. In the graphical notation, this attribute will be denoted by a corresponding image (\includegraphics[width=0.03\textwidth]{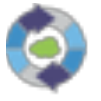},\includegraphics[width=0.03\textwidth]{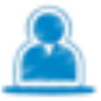},\includegraphics[width=0.03\textwidth]{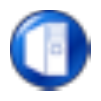} and \includegraphics[width=0.03\textwidth]{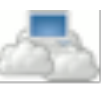}, resp.);
	\sort{ident},  models the identifier of the vertex;
	\sort{size}, denotes the size of the machine and virtual machine;
	\sort{free}, describes the amount of unused space in a machine;
	\sort{type}, describes the type of a virtual machine (as a simplification, we assumed that there is a set of standard virtual machines that may be requested by users, identified by their types);
	\sort{config}, this models the internal configuration of the cloud administrator, probably this would be a set of tables and variables describing the current state of machines and virtual machines;
\item[Edge attributes:]
	\sort{edgeType}, some arcs will represent physical relations (like a cloud administrator is connected to all machines monitored by it) or "knows"-relations (like a user may know a cloud administrator) and others will represent messages that are sent in the system. Messages will be denoted by dashed arrows, all other relations will be solid edges; 
	\sort{type}, analogous to the types of vertices;
	\sort{id}, used in messages that require a  parameter (identifier of a virtual machine). 
\end{description}	

The data types used in the state graph are defined in specification $Cloud\_Sp$ (Figure \ref{fig:cloudsp}). This specification includes sorts for booleans and natural numbers with usual operations and equations,  sort \sort{T} for the different types of virtual machines, and a sort \sort{C} to describe configurations of a cloud administrator. Such configurations are records containing the current status of the cloud. Due to space limitations, we will not define details of configurations, just use some basic operations (equations will be also omitted).
%\begin{description}
%\item \opn{newId:} \sort{C} $\times$ \sort{N} $\to$ \sort{B} checks whether an identifier is not used in a configuration
%\item \opn{enoughSpace: } \sort{C} $\times$ \sort{N} $\to$ \sort{B} checks whether there is enough space in a configuration
%\item \opn{newVM:} \sort{C} $\times$ \sort{N} $\times$ \sort{N} $\times$ \sort{N} $\times$ \sort{VMType} $\to$ \sort{C} includes a new virtual machine in a configuration
%\item \opn{replVM: }  \sort{C} $\times$ \sort{N} $\times$ \sort{N} $\to$ \sort{C} replicates a virtual machine in a configuration
%\item \opn{newMch:} \sort{C} $\times$ \sort{N} $\times$ \sort{N}  $\times$\sort{Nat} $\to$ \sort{C} includes a new machine in a configuration
%\item \opn{mergeMch:} \sort{C} $\times$ \sort{N} $\times$ \sort{N} $\to$ \sort{C} merges two machines in a configuration
%\end{description}

\begin{figure}[htp]
\centering
\framebox %[\textwidth]
{
\begin{minipage}{\textwidth}
\begin{tabbing}
\it \small
{Cloud\_Sp} : \\
\= {\bf sorts} \= \sort{B}, \sort{N}, \sort{C},\sort{T} \\
\> {\bf opns} \>   \  \hspace{5cm}                                                                     \= \\
\>                 \>  $\ldots$ \> \emph{boolean operators...}\\
\>                 \>  $\ldots$ \> \emph{natural numbers operators...}\\
\>                  \> \opn{newId:} \sort{C} $\times$ \sort{N} $\to$ \sort{B} 
 			\> \emph{checks whether an id is not used in a config}\\
\>                  \> \opn{enoughSpace: } \sort{C} $\times$ \sort{N} $\to$ \sort{B}
			\> \emph{checks if there is enough space in a config}\\
\>                  \>  \opn{newVM:} \sort{C} $\times$ \sort{N} $\times$ \sort{N} $\times$ \sort{N} $\times$ \sort{T} $\to$ \sort{C}
			\> \emph{includes a new virtual machine in a config}\\
\>                  \> 	 \opn{replVM: }  \sort{C} $\times$ \sort{N} $\times$ \sort{N} $\to$ \sort{C}  
			\> \emph{replicates a virtual machine in a config}\\
\>                  \> 	 \opn{newMch:} \sort{C} $\times$ \sort{N} $\times$ \sort{N}  $\times$\sort{Nat} $\to$ \sort{C}
			\> \emph{ includes a new machine in a config}\\
\> 		\> 	\opn{mergeMch:} \sort{C} $\times$ \sort{N} $\times$ \sort{N} $\to$ \sort{C} 
			\> \emph{merges two machines in a config}	\\		
\> 		\> 	\opn{replicateAdm?:} \sort{C} $\to$ \sort{B} 
			\> \emph{checks whether a new administrator is needed}	\\									
\> {\bf eqns} \\
\>                 \>  $\ldots$
\end{tabbing}
\end{minipage}
}
\caption{Specification $Cloud\_Sp$}
\label{fig:cloudsp}
\end{figure}

For example, the graph $\mathbf{G1}$ depicted in Fig.~\ref{fig:rules} describes two users and one cloud administrator that knows one machine, $M1$, and two types of virtual machines, $T1$ and $T2$. Actually, the administrator stores the images of the corresponding virtual machines such that, when a request is done, it creates a copy of this image in an available machine. Images are modelled by a special identifier ({\tt 0}). There are also two request messages, one from each user.

\begin{figure}[htbp]
\begin{center}
\includegraphics[width=1\textwidth]{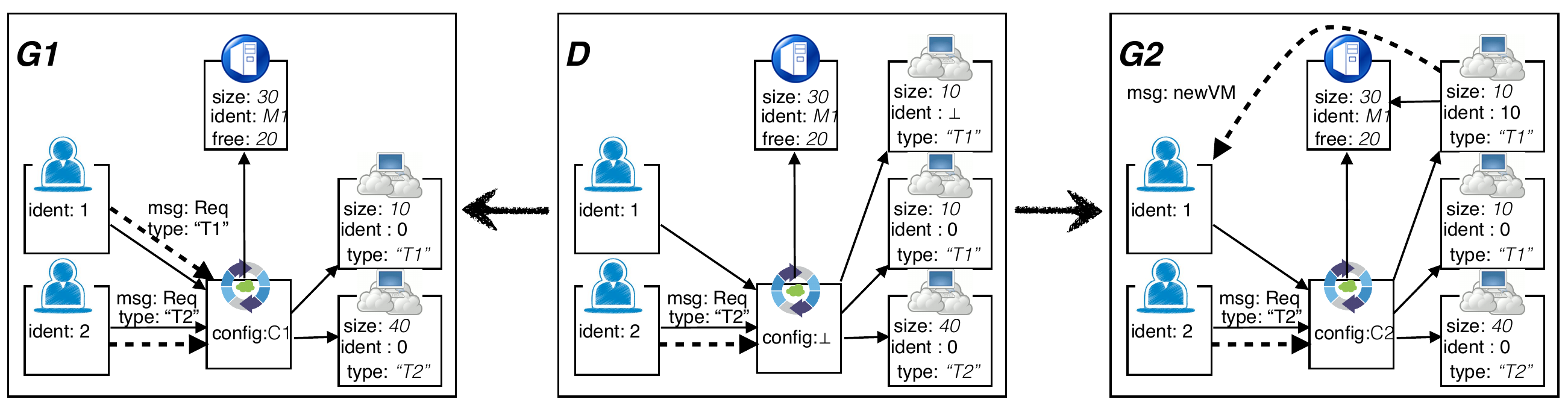}
\includegraphics[width=1\textwidth]{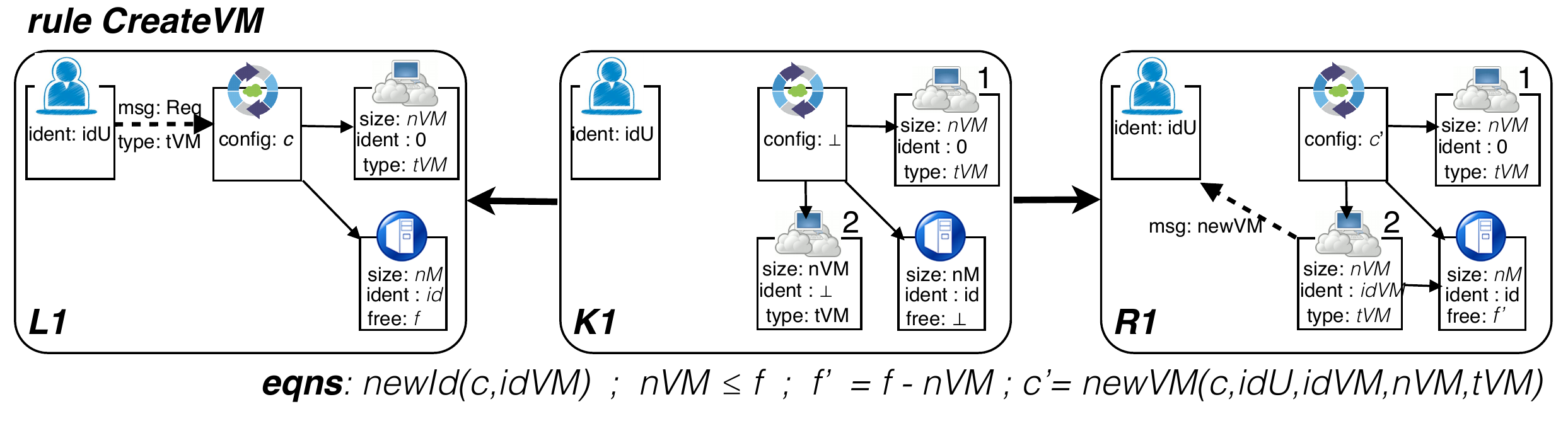}
\includegraphics[width=1\textwidth]{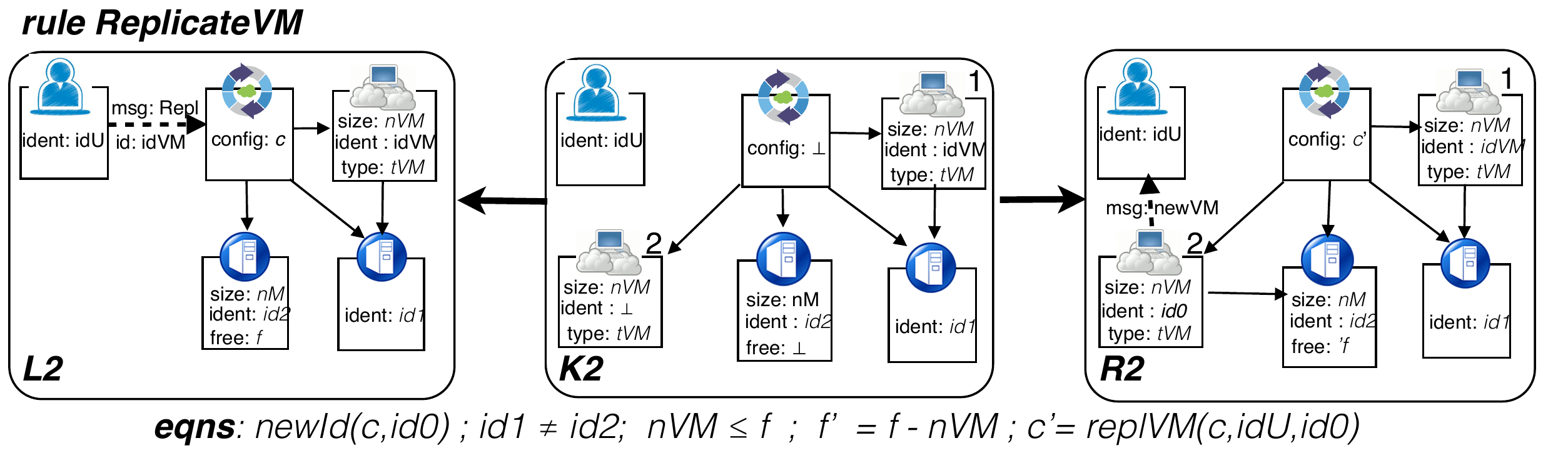}
\includegraphics[width=1\textwidth]{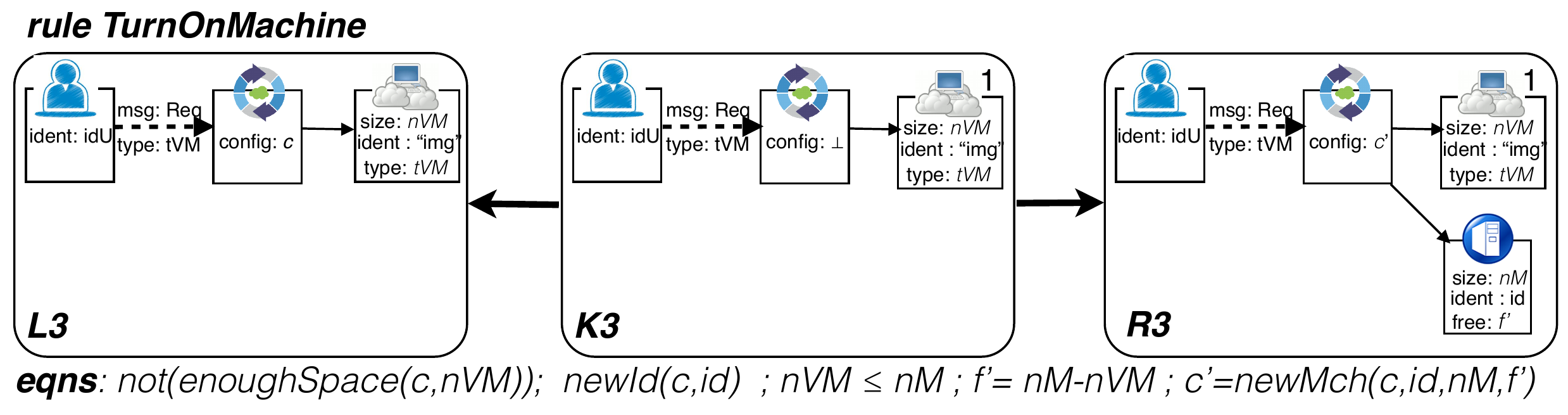}
\includegraphics[width=1\textwidth]{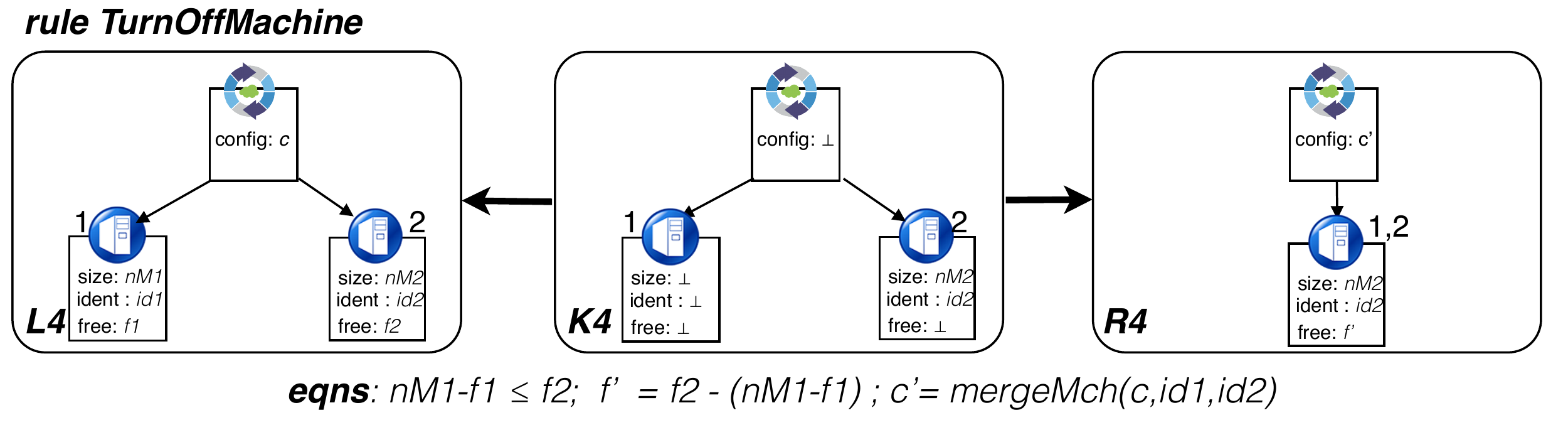}
\includegraphics[width=1\textwidth]{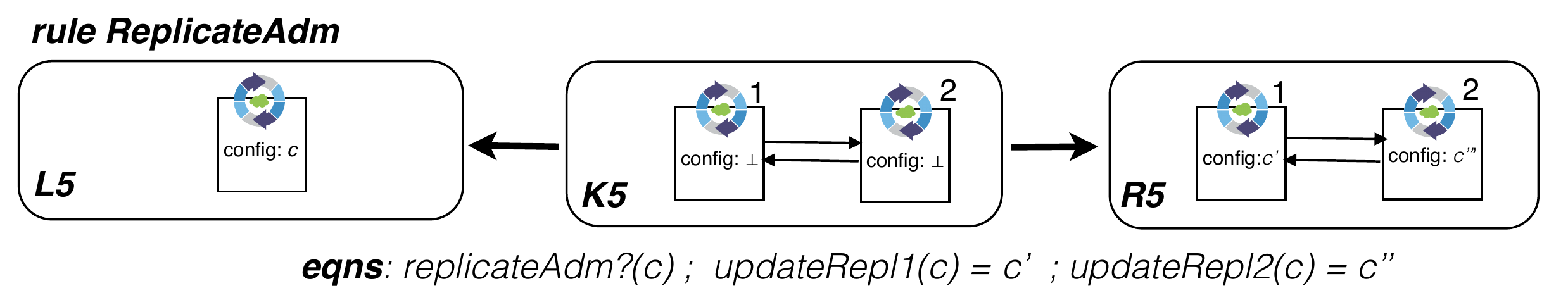}
\end{center}
\caption{Graph and Rules of the Cloud Administrator}
\label{fig:rules}
\end{figure}

\subsection{Cloud Administration: Dynamic Part}

 Figure  \ref{fig:rules} also shows some rules that describe the behaviour of the cloud administrator.  Rule \rulename{CreateVM} models the creation of a new virtual machine. This may happen when there is a request from a user (dashed edge in  $\mathbf{ L1}$) having as attribute the type of virtual machine that is created and the cloud administrator has a corresponding image and a machine to install this VM. Some additional constraints  over the attributes are modelled by equations (written below the rule): the identifier that will be used for the new VM is fresh (\opn{newId}$(c,idVM)$), there is enough free space in the chosen machine ($nVM\leq f$)\footnote{To enhance readability, when working with boolean expressions in equations, we omit the right side of the equation. For example, we write simply $newId(c,Id)$ instead of $newId(c,Id)=\sort{true}$.}. The remaining equations describe the values that some attributes will receive when this rule is applied: variable $f'$ depicts the amount of free space in the machine after the installation of the new VM, and $c'$ is the updated configuration of the could administrator. Note that the two instances of the VM in $\mathbf {K1}$ are copies of the corresponding vertex in $\mathbf {L1}$, just the identifier attribute in the second copy is left undefined, the attributes \opn{config} and \opn{free} are also undefined, since their values will change. Finally, in $\mathbf {R1}$, this second copy is updated with the new identifier ($idVM$) and it is installed in the machine and sent to the user, and the attributes of the cloud administrator and machine are updated accordingly. Application of this rule to graph $\mathbf{G1}$ is given by the span $G1 \leftarrow D \rightarrow G2$ on top of Fig.~\ref{fig:rules}.
 
 Rule \rulename{replicateVM} creates a copy (replica) of a VM in another physical machine. This operation is important for fault tolerance reasons. When this rule is applied, all references to the original VM will also point to the new VM. The configuration of the cloud administrator is updated because any change in one virtual machine must now be propagated to its copy. Rule \rulename{replicateAdm}  is used to replicate the cloud administrator itself. This kind of operation may be necessary, for example,  when the number of clients becomes too large or for dependability reasons. The rule that specifies the operation has an equation that checks whether this replication is needed ($replicateAdm?(c)$). In case this is true in the current configuration, the administrator is copied and the two configurations (the original and the copy) are updated (because now they must know that some synchronisation is needed to perform the operations). Since these are copies, they manage the same machines and VMs, but now clients may send requests to either of the administrators  (when this rule is applied, all edges that were connected to one administrator will also be connected to the copy).
 
 Rules \rulename{TurnOnMachine} and \rulename{TurnOffMachine} model the  creation and deletion of machines in the system. We assumed that there is an unlimited number of machines that may be connected to the system, and thus there is a need for more capacity ($not(enoughSpace(c,nVM)$ is true), a new machine may be added. We specified a simple version of turning off  a machine  by merging the vertices that correspond to two different machines. This can be done if the administrator notices that there is enough free space in one machine to accommodate VMs that are in another machines ($nM1-f1\leq f2$ is true).  When this rule is applied, all VMs that were in both machines will end up in the machine with identifier $id2$.

%-------------------------------------------------------------------------------
%-------------------------------------------------------------------------------
\vspace{-0.3cm}
\section{Related Work}
\label{sec:relatedwork}

 Various definitions of attributed graphs have been proposed in the
 literature. Labelled graphs, e.g. \cite{HabelP02}, in which attributes
 are limited to a simple set of a vocabulary, could be considered as a
 first step towards attributed graphs.  Such a set of vocabulary can
 be replaced by a specific, possibly infinite, set (of attributes) such
 as integers yielding particular definition of attributed graphs.
 This approach has been proposed for instance in \cite{PlumpS04} and
 could be considered as a particular case of the definition of
 attributed graphs we proposed in this paper.

 The most popular way to define the data part in attributed graphs is
 based on algebraic specifications, see
 e.g. \cite{LKW93,HeckelKT02,Berthold2002,EhrigEPT06}.  E-Graphs
 \cite{EhrigEPT06} is one of the principal contribution in this perspective, where an
 attributed graph gathers, in addition to its own vertices and edges,
 additional vertices and edges corresponding to the attribution part. The
 latter vertices correspond to possible attribution values. Such vertices
 might be infinite whenever the set of attributes is infinite.  An
 attribution edge goes from a vertex or an edge of the considered graph
 to an attribution vertex.  Attribution edges are used to represent
 graphically attribution functions.  Due to the representation of each
 attribute as a vertex, an E-graph is infinite in general.

To overcome the infinite structures of E-graphs, Symbolic graphs
\cite{OrejasL10} have been proposed. They are E-graphs which have
variables as attributes. Such variables can be constrained by means of
first order logic formulae. Hence a symbolic graph represents in
concise way a (possibly infinite) set of (ground) E-graphs. 

In this paper, we have proposed a general definition of attributed
structures where the data part is not necessarily specified as an
algebra. Our approach is very close to the recent paper by
U. Golas~\cite{Golas12} where an attributed graph is also defined as a
tuple $(G, A, \att)$ where $G$ is a given structure, $A$ consists of attribution values
and $\att$ is a \emph{family} of partial attribution functions. The
main difference with our proposal lies in the consideration of
attribution functions $\att$.  For sake of simplicity, we considered
simply partial functions for $\att$. Generalization to families of
functions as in \cite{Golas12} is straightforward.

Besides the variety of definitions of attributed graphs as mentioned
above, attributed graph transformation rules have been based mainly on
the double pushout approach which departs from the sesquipushout
approach we have used in our framework. For a comparison of the double
and the sesquipushout approaches we refer the reader to
\cite{CorradiniHHK06}. As far as we are aware of, the present paper
presents the first study of attributed graph transformations following
the sesquipushout approach and thus featuring the possibility of vertex
and edge cloning in presence of attributes. Thanks to partial
morphisms, rules allow also deletion and change of attributes.

%-------------------------------------------------------------------------------
%-------------------------------------------------------------------------------
\section{Conclusion} 
\label{sec:conclusion}

In this paper we presented an approach to transformations of
attributed structures that allows cloning and merging of items.  This
approach is based on the SqPO approach to graph
transformations, and thus also allows deletion in unknown
context. Concerning the attributes, our framework is general in the
sense that many different kinds of attributes can be used (not just
algebras, as in most attributed graph transformation definitions) and
allows that rules change the attributes associated to
vertices/edges. The resulting formalism is very interesting and we
believe that it can be used to provide suitable specifications of many
classes of applications like cloud computing, adaptive systems, and
other highly dynamically changing systems.

As future work, we plan to develop more case studies to understand the
strengths and weaknesses of this formalism for practical
applications. We also want to study analysis methods. Since we are
allowing non-injective rules, great part of the theory of graph
transformations can not be used directly and we need to investigate
which results may hold. Concerning verification of properties, we
intent to extend the analysis of graph transformations using theorem
provers \cite{RibeiroCosta} to attributed SqPO-rewriting.

\bibliographystyle{abbrv} \bibliography{main}

\appendix
\section*{Appendix}

%-------------------------------------------------------------------------------
% appendix

Theorem~\ref{theo:sqpo} about the lifting of sesqui-pushouts 
from the category of structures $\G$ to the 
category of attributed structures $\PAG$ relies on 
two Propositions which are stated and proved in this Appendix: 
Proposition~\ref{prop:po} about the lifting of pushouts 
in Section~\ref{app:po} 
and Proposition~\ref{prop:fpbc} about the lifting of final pullback complements 
in Section~\ref{app:fpbc}

%-------------------------------------------------------------------------------
\section{Lifting Pushouts to attributed structures}
\label{app:po}

\begin{remark}
\label{remark:po-sets}
In the category of sets, 
let us consider a commutative square as follows: 
$$ \xymatrix@C=6pc{ 
\ar@{}[rd]|{(=)}
X \ar[r]^{f} \ar[d]_{i} & X' \ar[d]^{i'} \\
Y \ar[r]^{f_1} & Y' \\
} $$ 
Let us assume that $i$ is an inclusion, so that $Y=X+C$ 
where $C$ is the complement of $X$ in $Y$. 
This square is a PO (up to isomorphism) if and only if
$i'$ is an inclusion, so that $Y'=X'+C'$ 
where $C'$ is the complement of $X'$ in $Y'$, 
and moreover $C'=C$ and $f_1=f+\id_C$:   
$$ \xymatrix@C=6pc{ \ar@{}[rd]|{(PO)}
X \ar[r]^{f} \ar[d]_{i} & X' \ar[d]^{i'} \\
Y=X+C \ar[r]^{f_1=f+\id_C} & Y'=X'+C \\
} $$ 
It follows that injections are stable under POs in $\Set$. 
\end{remark} 

\begin{proposition}
\label{prop:po}
Let us assume that the functor $S:\G\to\Set$ preserves POs.
Then the functor $U_\G:\PAG\to\G$ lifts POs 
of matchings along attribute-preserving morphisms, 
in the following precise way.
Let $(r,\id_A):\ag{K}\to\ag{R}$ be a morphism in $\PAG$.
Let $(m_K,a):\ag{K}\to\ag{D}$ be a matching,
so that we can denote $SD=SK+C$ with $Sm_K$ the canonical inclusion.
Let $\Delta_{\po}$ (below on the left) be a PO in $\G$, 
then % since S preserves POs 
we can denote $SH=SR+C$ with $Sm_R$ the canonical inclusion
and $Sr_1=Sr+\id_C$.
Let $\att_H:SH\pto JA_1$ be defined as $Ja\circ \att_R$ on $SR$ 
and as $\att_D$ on $C$. 
Then $(m_R,a)$ is a matching, 
$(r_1,\id_{A_1})$ is a morphism in $\PAG$  
and $\ag{\Delta}_{\po}$ (below on the right) is a PO in $\PAG$. 
$$ \begin{array}{lll}
\Delta_{\po}: & \quad & \ag{\Delta}_{\po}: \\ 
\xymatrix@C=4pc{  \ar@{}[rd]|{(PO)}
K \ar[r]^{r} \ar[d]_{m_K} & 
  R \ar[d]^{m_R} \\
D \ar[r]^{r_1} & 
  H \\
} &&
\xymatrix@C=4pc{  \ar@{}[rd]|{(PO)}
(K,A,\att_K) \ar[r]^{(r,\id_A)} \ar[d]_{(m_K,a)} & 
  (R,A,\att_R) \ar[d]^{(m_R,a)} \\
(D,A_1,\att_D) \ar[r]^{(r_1,\id_{A_1})} & 
  (H,A_1,\att_H)  \\
} \\ 
\end{array} $$ 
\end{proposition}

\begin{proof} 
We have:
  \begin{itemize}
  \item $SD=SK+C$ with coprojections $Sm_K:SK\to SD$ and (say) $i_D:C\to SD$, 
  \item $SH=SR+C$ with coprojections $Sm_R:SR\to SH$ and (say) $i_H:C\to SH$, 
  \item $Sr_1$ is characterized by 
  $Sr_1\circ i_D = i_H$ and $Sr_1\circ Sm_K= Sm_R\circ Sr$,
  \item $\att_H$ is defined by 
  $\att_H \circ i_H = \att_D \circ i_D$ 
  and $\att_H \circ Sm_R = Ja\circ \att_R$. 
  \end{itemize}

We have to check the following properties.
\begin{itemize}

\item $(m_R,a)$ is a matching in $\PAG$. 
Indeed, $\att_H\circ Sm_R = Ja \circ \att_R:SR\to JA_1$ 
by definition of $\att_H$ and $Sm_R$ is an injection by assumption.

\item $(r_1,\id_{A_1})$ is a morphism in $\PAG$.
This means that $\att_H\circ Sr_1 \geq \att_D : SD\to JA_1$, 
which is equivalent
to $\att_H\circ Sr_1 \circ i_D \geq \att_D \circ i_D : C\to JA_1$ 
and $\att_H\circ Sr_1 \circ Sm_K \geq \att_D \circ Sm_K: SK \to JA_1$. 
  \begin{itemize}
  \item On $i_D(C)$ we have 
  $\att_H \circ Sr_1\circ i_D = \att_H \circ i_H = \att_D \circ i_D$.
  \item On $m_K(SK)$ we have 
  $\att_H \circ Sr_1\circ Sm_K = \att_H \circ Sm_R \circ Sr 
  = Ja \circ \att_R \circ Sr$.
  We know that $\att_R \circ Sr \geq \att_K$ 
  because $(r,\id_A)$ is a morphism in $\PAG$, 
  thus we get $\att_H \circ Sr_1\circ Sm_K \geq Ja \circ \att_K$.
  And we know that $\att_D \circ Sm_K  = Ja \circ \att_K$ 
  because $(m_K,a)$ is a \emph{strict} morphism in $\PAG$, 
  so that we get   $\att_H \circ Sr_1\circ Sm_K \geq \att_D \circ Sm_K$.
  \end{itemize}

\item $\ag{\Delta}_{\po}$ is a commutative square in $\PAG$.
This is obvious, since both its underlying square in $\G$ 
and its underlying square in $\A$ are  commutative. 

\item $\ag{\Delta}_{\po}$ is a PO in $\PAG$.
Let us consider a commutative square in $\PAG$: 
$$ \xymatrix@C=6pc{  \ar@{}[rd]|{(=)}
(K,A,\att_K) \ar[r]^{(r,\id_A)} \ar[d]_{(m_K,a)} & 
  (R,A,\att_R) \ar[d]^{(m',a')} \\
(D,A_1,\att_D) \ar[r]^{(r',b))} & 
  (H',A',\att')  \\
} $$ 
and let us look for a morphism 
$(f,c): (H,A_1,\att_H)\to (H',A',\att')$ in $\PAG$
such that $(f,c)\circ (m_R,a)=(m',a')$ 
and $(f,c)\circ (r_1,\id_{A_1})=(r',b)$.
$$ \xymatrix@C=6pc{
(K,A,\att_K) \ar[r]^{(r,\id)} \ar[d]_{(m_K,a)} & 
  (R,A,\att_R) \ar[d]^(.6){(m_R,a)} \ar[ddr]^{(m',a')}  & \\
(D,A_1,\att_D) \ar[r]^{(r_1,\id)} \ar[drr]_{(r',b)} & 
  (H,A_1,\att_H) \ar@{-->}[dr]^{(f,c)} & \\
& & (H',A',\att') \\ 
} $$
Thus, we are looking for:
  \begin{itemize}
  \item A morphism 
  $f: H\to H'$  in $\G$ such that 
  $f\circ m_R=m'$ and $f\circ r_1=r'$:
  there is exactly one choice for $f$ 
  since $\Delta_{\po}$ is a PO in $\G$.
  \item A morphism 
  $c: A_1\to A'$  in $\A$ such that 
  $c\circ a=a'$ and $c\circ\id_{A_1}=b$:
  there is exactly one choice for $c$ (since $b\circ a=a'$),
  it is $c=b$. 
  \end{itemize}
Now we have to check that $\att' \circ Sf \geq Jb \circ \att_H : SH\to JA'$. 
This is equivalent
to $\att' \circ Sf \circ i_H \geq Jb \circ \att_H \circ i_H : C\to JA'$  
and $\att' \circ Sf \circ Sm_R \geq Jb \circ \att_H \circ Sm_R : SR\to JA'$. 
  \begin{itemize}
  \item Since $i_H = Sr_1\circ i_D $ and $f\circ r_1=r'$ 
  we have $\att' \circ Sf \circ i_H = 
  \att' \circ Sf \circ Sr_1 \circ i_D =
  \att' \circ Sr' \circ i_D $. 
  But $\att' \circ Sr' \geq Jb \circ \att_D$ because 
  $(r',b)$ is a morphism in $\PAG$, 
  thus $\att' \circ Sf \circ i_H \geq Jb \circ \att_D \circ i_D $. 
  The definition of $\att_H$ states that $\att_D \circ i_D = \att_H \circ i_H$, 
  so that  $\att' \circ Sf \circ i_H \geq Jb \circ \att_H \circ i_H$. 
  \item Since $f\circ m_R=m'$ 
  we have $\att' \circ Sf \circ Sm_R = \att' \circ Sm'$. 
  But $\att' \circ Sm' \geq Ja' \circ \att_R$ because  
  $(m',a')$ is a morphism in $\PAG$, 
  thus  $\att' \circ Sf \circ Sm_R  \geq Ja' \circ \att_R$. 
  Since $a'=b\circ a$ we get 
  $\att' \circ Sf \circ Sm_R  \geq Jb \circ Ja \circ \att_R$. 
  And we know that $\att_H \circ Sm_R  = Ja \circ \att_R$ 
  because $(m_R,a)$ is a \emph{strict} morphism in $\PAG$, 
  so that we get $\att' \circ Sf \circ Sm_R \geq Jb \circ \att_H \circ Sm_R$. 
  \end{itemize}

\end{itemize}
\end{proof} 

Let us summarize what may occur in $\ag{\Delta}_{\po}$ 
for an element $x\in SD$, using the previous conventions for notations
(here $t$ denotes an element of $TA$ and $t_1$ an element of $TA_1$).
If $x\in C$ then two cases may occur: 
$$
\xymatrix@R=1pc@C=1pc{
x:t_1 \ar@{|->}[r] & x:t_1 \\ 
}
\qquad 
\xymatrix@R=1pc@C=1pc{
x:\bot \ar@{|->}[r] & x:\bot \\ 
}
$$
If $x\in SK$ then three cases may occur: 
$$
\xymatrix@R=1pc@C=1pc{
x:t \ar@{|->}[r] \ar@{|->}[d] & r(x):t \ar@{|->}[d] \\ 
x:a(t) \ar@{|->}[r] & r(x):a(t) \\ 
}
\qquad 
\xymatrix@R=1pc@C=1pc{
x:\bot \ar@{|->}[r] \ar@{|->}[d] & r(x):t \ar@{|->}[d] \\ 
x:\bot \ar@{|->}[r] & r(x):a(t) \\ 
}
\qquad 
\xymatrix@R=1pc@C=1pc{
x:\bot \ar@{|->}[r] \ar@{|->}[d] & r(x):\bot \ar@{|->}[d] \\ 
x:\bot \ar@{|->}[r] & r(x):\bot \\ 
}
$$

%-------------------------------------------------------------------------------
\section{Lifting Final Pullback Complements to attributed structures}
\label{app:fpbc}

\begin{remark}
\label{remark:pb-sets}
In the category of sets, 
let us consider a commutative square as follows: 
$$ \xymatrix@C=6pc{ 
X \ar[d]_{i} & X' \ar[l]_{f} \ar[d]^{i'} \\
Y & Y' \ar[l]_{f_1} \\
} $$ 
Let us assume that $i$ is an inclusion, so that $Y=X+C$. 
This square is a PB (up to isomorphism) if and only if
$i'$ is an inclusion, so that $Y'=X'+C'$, 
and $f_1=f+\gamma$ for some $\gamma:C'\to C$:   
$$ \xymatrix@C=6pc{ \ar@{}[rd]|{(PB)}
X \ar[d]_{i} & X' \ar[l]_{f} \ar[d]^{i'} \\
Y=X+C & Y'=X'+C' \ar[l]_{f_1=f+\gamma} \\
} $$ 
It follows that injections are stable under PBs in $\Set$, 
which is a special instance of the well-known fact that 
monomorphisms are stable under PBs in all categories. 
Moreover, this PB is a FPBC (up to isomorphism) if and only if
$C'=C$ and $\gamma=\id_C$:   
$$ \xymatrix@C=6pc{ \ar@{}[rd]|{(FPBC)}
X \ar[d]_{i} & X' \ar[l]_{f} \ar[d]^{i'} \\
Y=X+C & Y'=X'+C \ar[l]_{f_1=f+\id_C} \\
} $$ 
\end{remark} 

\begin{lemma}[PBs in $\PAG$]
\label{lemm:pb}
Let us assume that the functors 
$U_{\G}:\PAG\to\G$, $U_{\A}:\PAG\to\A$, 
$S:\G\to\Set$ and $J:\A\to\Set$ preserve PBs. 
Let us consider a PB in $\PAG$: 
$$ \xymatrix{  \ar@{}[rd]|{(PB)}
(L,A,\att_L) \ar[d]_{(m_L,a)} & 
  (K,A',\att_K) \ar[l]_{(l,b)} \ar[d]^{(m_K,a')} \\
(G,A_1,\att_G) & (D,A'_1,\att_D) \ar[l]_{(l_1,b_1)}
}$$
If $(m_L,a)$ is a matching then $(m_K,a')$ is a matching 
and $(Sl_1)^{-1}(Sm_L(SL)) \subseteq Sm_K(SK)$.
\end{lemma}

\begin{proof} 
Since $S\circ U_{\G}:\PAG\to\Set$ preserves PBs, 
it follows from Remark~\ref{remark:pb-sets} 
that $m_K$ is injective and that $(Sl_1)^{-1}(Sm_L(SL)) \subseteq Sm_K(SK)$. 
It remains to prove that $(m_K,a')$ is strict. 
% to improve?... ***
If this does not hold, then there is some 
$x\in SK$ such that $x\not\in\D(\att_K)$ and $Sm_K(x)\in\D(\att_D)$. 
Let $y=Sm_K(x)$, $x_0=Sl(x)$ and $y_0=l_1(y)=m_L(x_0)$. 
Since $y\in\D(\att_D)$, 
$(l_1,b_1)$ is a morphism and $(m_L,a)$ a strict morphism, 
we have $y_0\in\D(\att_G)$ and $x_0\in\D(\att_L)$.
More precisely, let $y:u$ and $x_0:t_0$, 
then $y_0:u_0$ where $u_0=a(t_0)=b_1(u)$. 
Since $J\circ U_{\A}:\PAG\to\Set$ preserves PBs, 
there is a unique $t\in JA'$ such that $t_0=Jb(t)$ and $u=Ja'(t)$. 
Let us define $\ag{K'}=(K,A',\att'_K)$ 
where $\att_K'$ extends $\att_K$ simply by mapping $x$ to $t$.
Then $b$ and $a'$ may be extended accordingly. 
The resulting square is commutative in $\PAG$  
but there is no morphism from $\ag{K'}$ to $\ag{K}$ in $\PAG$, 
which contradicts the fact that the given square is a PB in $\PAG$.
\end{proof} 

\begin{proposition}
\label{prop:fpbc}
Let us assume that the functors 
$U_{\G}:\PAG\to\G$, $U_{\A}:\PAG\to\A$, $S:\G\to\Set$ and $J:\A\to\Set$
preserve PBs. 
% S does not preserve FPBC: cf "Edge" 
Then the functor $U_\G:\PAG\to\G$ lifts FPBCs 
of matchings along attribute-preserving morphisms, 
in the following precise way.
Let $(l,\id_A):\ag{K}\to\ag{L}$ be a morphism in $\PAG$.
Let $(m_L,a):\ag{L}\to\ag{G}$ be a matching,
so that we can denote $SG=SL+C$ with $Sm_L$ the canonical inclusion.
Let $\Delta_{\fpbc}$ (below on the left) be a FPBC in $\G$, 
then % since S preserves PBs 
we can denote $SD=SK+C_1$ with $Sm_K$ the canonical inclusion
and $Sl_1=Sl+\gamma$ for some $\gamma:C_1\to C$, 
by Remark~\ref{remark:pb-sets}. 
Let $\att_D:SD\pto JA_1$ be defined as $Ja\circ \att_K$ on $SK$ 
and as $\att_G\circ Sl_1$ on $C_1$. 
Then $(m_K,a)$ is a matching, 
$(l_1,\id_A)$ is a morphism in $\PAG$  
and $\ag{\Delta}_{\fpbc}$ (below on the right) is a FPBC in $\PAG$.
$$ \begin{array}{lll}
\Delta_{\fpbc}: & \quad & \ag{\Delta}_{\fpbc}: \\ 
\xymatrix@C=4pc{  \ar@{}[rd]|{(FPBC)}
L \ar[d]_{m_L} & 
  K \ar[l]_{l} \ar[d]^{m_K} \\
G & D \ar[l]_{l_1}
} &&
\xymatrix@C=4pc{  \ar@{}[rd]|{(FPBC)}
(L,A,\att_L) \ar[d]_{(m_L,a)} & 
  (K,A,\att_K) \ar[l]_{(l,\id_A)} \ar[d]^{(m_K,a)} \\
(G,A_1,\att_G) & (D,A_1,\att_D) \ar[l]_{(l_1,\id_{A_1})}
} \\ 
\end{array} $$ 
\end{proposition}

\begin{proof} 
We have: 
  \begin{itemize}
  \item $SD=SK+C_1$ with coprojections $Sm_K:SK\to SD$ 
  and (say) $i_D:C_1\to SD$, 
  \item $\att_D$ is defined by 
  $\att_D \circ i_D = \att_G \circ Sl_1 \circ i_D$ 
  and $\att_D \circ Sm_K = Ja \circ \att_K$. 
  \end{itemize}

We have to check the following properties.
\begin{itemize}

\item $(m_K,a)$ is a matching in $\PAG$. 
Indeed, $\att_D\circ Sm_K = Ja \circ \att_K:SK\to JA_1$
by definition of $\att_D$ and $Sm_K$ is an injection by assumption.

\item $(l_1,\id_{A_1})$ is a morphism in $\PAG$. 
This means that $\att_G\circ Sl_1 \geq \att_D : SD\to JA_1$, 
which is equivalent
to $\att_G\circ Sl_1 \circ i_D \geq \att_D \circ i_D : C_1\to JA_1$ 
and $\att_G\circ Sl_1 \circ Sm_K \geq \att_D \circ Sm_K: SK \to JA_1$.
  \begin{itemize} 
  \item On $i_D(C_1)$ we have  
  $\att_G \circ Sl_1\circ i_D = \att_D \circ i_D$.
  \item On $m_K(SK)$ we have 
  $\att_G \circ Sl_1\circ Sm_K = \att_G \circ Sm_L \circ Sl$ 
  since the square $\Delta_{\fpbc}$ is commutative.
  We know that $\att_G \circ Sm_L  \geq Ja \circ \att_L$ 
  because $(m_K,a)$ is a morphism in $\PAG$, % maybe NOT strict ***
  and that $\att_L \circ Sl \geq \att_K$ 
  because $(l,\id_A)$ is a morphism in $\PAG$, 
  thus we get $\att_G \circ Sm_L \circ Sl
  \geq Ja \circ \att_L \circ Sl
  \geq Ja \circ \att_K$,
  so that $\att_G \circ Sl_1\circ Sm_K  \geq Ja \circ \att_K$. 
  From the definition of $\att_D$ on $m_K(SK)$ we get 
  $\att_G \circ Sl_1\circ Sm_K  \geq \att_D \circ Sm_K$. 
  \end{itemize}

\item $\ag{\Delta}_{\fpbc}$ is a commutative square in $\PAG$. 
This is obvious, since both its underlying square in $\G$ 
and its underlying square in $\A$ are commutative. 

\item $\ag{\Delta}_{\fpbc}$ is a PB in $\PAG$.
Let us consider a commutative square in $\PAG$: 
$$\xymatrix@C=6pc{  \ar@{}[rd]|{(=)}
(L,A,\att_L) \ar[d]_{(m_L,a)} & 
  (K',A',\att') \ar[l]_{(l',b)} \ar[d]^{(m',a')} \\
(G,A_1,\att_G) & (D,A_1,\att_D) \ar[l]_{(l_1,\id_{A_1})}
}$$
and let us look for a morphism 
$(f,c): (K',A',\att')\to (K,A,\att_K)$ in $\PAG$
such that $(m_K,a)\circ(f,c) =(m',a')$ 
and $(l,\id_A)\circ(f,c) =(l',b)$.
$$ \xymatrix@C=6pc{
& & (K',A',\att') \ar[lld]_{(l',b)} \ar@{-->}[ld]^{(f,c)} \ar[ldd]^{(m',a')} \\ 
(L,A,\att_L) \ar[d]_{(m_L,a)} & 
  (K,A,\att_K) \ar[l]^{(l,\id)} \ar[d]^(.4){(m_K,a)} & \\
(G,A_1,\att_G) & 
  (D,A_1,\att_D) \ar[l]_{(l_1,\id)} & \\
} $$
Thus, we are looking for:
  \begin{itemize}
  \item A morphism 
  $f: K'\to K$  in $\G$ such that 
  $m_K\circ f=m'$ and $l\circ f=l'$:
  there is exactly one choice for $f$ 
  since $\Delta_{\fpbc}$ is a PB in $\G$. 
  \item A morphism 
  $c: A'\to A$  in $\A$ such that 
  $a\circ c=a'$ and $\id_A\circ c=b$:
  there is exactly one choice for $c$ (since $a\circ b=a'$),
  it is $c=b$. 
  \end{itemize}
Now we have to check that $\att_K \circ Sf \geq Jb \circ \att' : SK'\to JA$. 
Let us denote $\varphi=\att_K \circ Sf$ and $\varphi'=Jb \circ \att'$.
  \begin{itemize}
  \item 
  % common upper bound 
  Let $\varphi''=\att_L \circ Sl'$ and let us prove that 
  $\varphi'' \geq \varphi$ and $\varphi'' \geq \varphi'$. 
  Since $(l,\id_A)$ is a morphism in $\PAG$ 
  we have $\att_L \circ Sl \geq \att_K$,
  and since $l'=l\circ f$ this implies that 
  $\varphi''=\att_L \circ Sl' \geq \att_K\circ Sf = \varphi$.
  Since $(l',b)$ is a morphism in $\PAG$ 
  we have $\varphi''=\att_L \circ Sl' \geq Jb \circ \att' = \varphi'$. 
  Thus, there is a partial function $\varphi''$ 
  such that $\varphi'' \geq \varphi$ and $\varphi'' \geq \varphi'$.
  It follows that $\varphi$ and $\varphi'$ coincide on 
  $\D(\varphi') \cap \D(\varphi)$. 
  So, if we can prove that $\D(\varphi') \subseteq \D(\varphi)$ 
  then we will get $\varphi \geq \varphi'$, as required. 
  \item 
  % domains 
  Now, let us prove that $Ja \circ \varphi \geq Ja \circ \varphi'$. 
  Since $(m_K,a)$ is a \emph{strict} morphism in $\PAG$ 
  we have $\att_D \circ Sm_K  = Ja \circ \att_K$, 
  so that $Ja \circ \varphi = Ja \circ \att_K \circ Sf 
  = \att_D \circ Sm_K \circ Sf $. 
  Since $(m',a')$ is a morphism in $\PAG$ 
  we have $\att_D \circ Sm' \geq Ja' \circ \att'$. 
  With $m'=m_K \circ f$ and $a'=a\circ b$ we get 
  $\att_D \circ Sm_K \circ Sf \geq Ja \circ Jb \circ \att'
  = Ja \circ \varphi'$. 
  Thus, we have $Ja \circ \varphi \geq Ja \circ \varphi'$. 
  This implies that $\D(Ja \circ \varphi') \subseteq \D(Ja \circ \varphi)$,
  and since $Ja$ is a total function 
  this means that $\D(\varphi') \subseteq \D(\varphi)$. 
  It follows that $\varphi \geq \varphi'$, as explained above.
  \end{itemize}

\item $\ag{\Delta}_{\fpbc}$ is a FPBC in $\PAG$.
Let us consider a PB in $\PAG$: 
$$\xymatrix@C=6pc{  \ar@{}[rd]|{(PB)}
(L,A,\att_L) \ar[d]_{(m_L,a)} & 
  (K',A',\att'_K) \ar[l]_{(l',b)} \ar[d]^{(m',a')} \\
(G,A_1,\att_G) & 
  (D',A'_1,\att'_D) \ar[l]_{(l'_1,b_1)}
}$$
with a morphism 
$(f,c): (K',A',\att'_K)\to (K,A,\att_K)$ in $\PAG$
such that $(l,\id_A)\circ(f,c) =(l',b)$, 
which implies that $c=b$.
And let us look for a morphism 
$(f_1,c_1): (D',A'_1,\att'_D)\to (D,A_1,\att_D)$ in $\PAG$
such that $(f_1,c_1)\circ(m',a') = (m_K,a)\circ(f,c) $ 
and $(l_1,\id_{A_1})\circ(f_1,c_1)=(l'_1,b_1)$.
$$ \xymatrix@C=6pc{
& & (K',A',\att'_K) \ar[lld]_{(l',b)} \ar[ld]^{(f,c)} \ar[ddd]^{(m',a')} \\ 
(L,A,\att_L) \ar[d]^{(m_L,a)} & 
  (K,A,\att_K) \ar[l]_{(l,\id)} \ar[d]^{(m_K,a)} & \\
(G,A_1,\att_G) & 
  (D,A_1,\att_D) \ar[l]_{(l_1,\id)} & \\
& & (D',A'_1,\att'_D) \ar[llu]^{(l'_1,b_1)} \ar@{-->}[lu]_{(f_1,c_1)} \\ 
} $$
Thus, we are looking for:
  \begin{itemize}
  \item A morphism 
  $f_1: D'\to D$  in $\G$ such that 
  $f_1\circ m' =m_K\circ f$ and $l_1\circ f_1=l'_1$:
  there is exactly one choice for $f_1$ 
  since $\Delta_{\fpbc}$ is a FPBC in $\G$.
  \item A morphism 
  $c_1: A'_1\to A_1$ in $\A$ such that 
  $c_1\circ a'=a\circ c$ and $\id_{A_1}\circ c_1=b_1$:
  there is exactly one choice for $c_1$ (since $b_1\circ a'=a\circ b$
  and $c=b$), it is $c_1=b_1$. 
  \end{itemize}
Now we have to check that 
$\att_D \circ Sf_1 \geq Jb_1 \circ \att'_D : SD'\to JA_1$. 
Let us denote $\psi=\att_D \circ Sf_1$ and $\psi'=Jb_1 \circ \att'_D$.
  \begin{itemize}
  \item
  % common upper bound 
  Let $\psi''=\att_G \circ Sl'_1$ and let us prove that 
  $\psi'' \geq \psi$ and $\psi'' \geq \psi'$. 
  Since $(l_1,\id_{A_1})$ is a morphism in $\PAG$ 
  we have $\att_G \circ Sl_1 \geq \att_D$,
  and since $l'_1=l_1\circ f_1$ this implies that 
  $\psi''=\att_G \circ Sl'_1 \geq \att_D\circ Sf_1 = \psi$.
  Since $(l'_1,b_1)$ is a morphism in $\PAG$ 
  we have $\psi''=\att_G \circ Sl'_1 \geq Jb_1 \circ \att'_D = \psi'$. 
  Thus, there is a partial function $\psi''$ 
  such that $\psi'' \geq \psi$ and $\psi'' \geq \psi'$.
  It follows that $\psi$ and $\psi'$ coincide on 
  $\D(\psi') \cap \D(\psi)$. 
  So, if we can prove that $\D(\psi') \subseteq \D(\psi)$  
  then we will get $\psi \geq \psi'$, as required. 
  \item 
  Now, let us prove that $\D(\psi') \subseteq \D(\psi)$,
  i.e., $\D(Jb_1 \circ \att'_D) \subseteq \D(\att_D \circ Sf_1)$.
  Since $Jb_1$ and $Sf_1$ are total functions, 
  we have to prove that for each $y'\in SD',\;$ 
  $ y'\in\D(\att'_D) \implies Sf_1(y')\in\D(\att_D)$. 
  Thus, let us consider some $y'\in\D(\att'_D)$, let $y=Sf_1(y')$,
  and let us prove that $y\in\D(\att_D)$. 
  Let $y_0=Sl_1(y)=Sl'_1(y')$, since $(l'_1,b_1)$ is a morphism in $\PAG$
  and $y'\in\D(\att'_D)$ we have $y_0\in\D(\att_G)$.  
  % below: diagrams? ***
    \begin{itemize}
    \item If $y'=Sm'(x')$ for some element $x'$ of $SK'$, 
    then let $x_0=Sl'(x')$, so that $y_0=Sl'_1(y')=S(l'_1\circ m')(x')=
    S(m_L\circ l')(x') = Sm_L(x_0)$. 
    Since $(m_L,a)$ is a \emph{strict} morphism in $\PAG$ 
    and $y_0\in\D(\att_G)$ we have $x_0\in\D(\att_L)$. 
    In addition, $S(m_K\circ f)(x')=S(f_1\circ m')(x')=Sf_1(y')=y$. 
    Since we consider a PB in $\PAG$,
    Lemma~\ref{lemm:pb} states that the morphism $(m',a')$ is strict. 
    Thus, since $y'\in\D(\att'_D)$, we have $x'\in\D(\att'_K)$. 
    Since $x'\in\D(\att'_K)$ and 
    since $(f,b)$ and $(m_K,a)$ are morphisms in $\PAG$ 
    we have $y=S(m_K\circ f)(x')\in\D(\att_D)$, as required.  
    \item Otherwise, $y'$ is not in the image of $Sm'$. 
    Since we consider a PB in $\PAG$,  
    Lemma~\ref{lemm:pb} states that $(Sl'_1)^{-1}(Sm_L(SL)) \subseteq Sm'(SK')$.
    Thus, since $y'\not\in Sm'(SK')$, we have $y_0\not\in Sm_L(SL)$, 
    i.e., $y_0\in C_1$.  
    Then, according to the definition of $\att_D$, 
    we have $y\in\D(\att_D)$, as required. 
    \end{itemize}
  Thus, we have proved that $\D(\psi') \subseteq \D(\psi)$, 
  which implies that $\psi \geq \psi'$, as explained above.
  \end{itemize}
\end{itemize}
\end{proof} 

Let us summarize what may occur in $\ag{\Delta}_{\fpbc}$ 
for an element $x\in SD$ 
(here $t$ denotes an element of $TA$ and $t_1$ an element of $TA_1$).
If $x\in C_1$ then two cases may occur: 
$$
\xymatrix@R=1pc@C=1pc{
x:t_1 & x:t_1 \ar@{|->}[l] \\ 
}
\qquad 
\xymatrix@R=1pc@C=1pc{
x:\bot & x:\bot \ar@{|->}[l] \\ 
}
$$
If $x\in SK$ then three cases may occur: 
$$
\xymatrix@R=1pc@C=1pc{
l(x):t \ar@{|->}[d] & x:t \ar@{|->}[l] \ar@{|->}[d] \\ 
l(x):a(t) & x:a(t) \ar@{|->}[l] \\ 
}
\qquad 
\xymatrix@R=1pc@C=1pc{
l(x):t \ar@{|->}[d] & x:\bot \ar@{|->}[l] \ar@{|->}[d] \\ 
l(x):a(t) & x:\bot \ar@{|->}[l] \\ 
}
\qquad 
\xymatrix@R=1pc@C=1pc{
l(x):\bot \ar@{|->}[d] & x:\bot \ar@{|->}[l] \ar@{|->}[d] \\ 
l(x):\bot & x:\bot \ar@{|->}[l] \\ 
}
$$

\section{Graph transformations with attributes defined equationally}

In this appendix we check that the assumptions of
Theorem~\ref{theo:sqpo} are satisfied in the case of attributed graph
transformations with attributes built over equational
specifications. The functors $U_{\G}$ and $U_{\A}$ are known from
Definition~\ref{defi:under}, we have to choose the categories $\G$ and
$\A$ and the functors $S$ and $T$.
%% \G 
Let $\G=\Gr$ be the category of graphs.
%% S 
Let $S:\G\to\Set$ be the functor which maps each graph 
to the disjoint union of its set of vertices and its set of edges. 
%% (and which is defined accordingly on graph morphisms). 
%% \A
Let $\Sp$ be some fixed equational specification 
and let $\A=\Mod(\Sp)$ be the category of models of $\Sp$.
%% T
Let $T:\A\to\Set$ be the functor which maps each model of $\Sp$ 
to the disjoint union of the carriers sets $A_s$ for $s$ in 
some given set of sorts.
%% (and which is defined accordingly on morphisms of models of $\Sp$). 

\begin{proposition}
With the definitions as above, 
the functors $U_{\G}:\PAG\to\Gr$, $U_{\A}:\PAG\to\Mod(\Sp)$, 
$S:\Gr\to\Set$ and $T:\Mod(\Sp)\to\Set$ preserve pullbacks,  
and $S$ preserves pushouts.
\end{proposition}

\begin{proof}
% G 
The category of graphs is the functor category $\Func(\C_{\gr},\Set)$ 
where $\C_{\gr}$ is the following small category:
$$ \xymatrix@C=6pc{
v \ar@(ul,dl)_{\id_v} & e \ar@(ur,dr)^{\id_e} \ar@/^/[l]^{t} \ar@/_/[l]_{s} \\ 
} $$
It follows that the category $\Gr$ has limits and colimits, 
in particular POs and PBs. 
% S  
The functors $V,E:\Gr\to\Set$ map each graph $G$
to its set of vertices $V(G)$ and its set of edges $E(G)$, respectively, 
and the functor $S:\Gr\to\Set$ is $S=V+E$, 
which maps each graph $G$ to the disjoint union $S(G)=V(G)+E(G)$.
The functors $V$ and $E$ are the \emph{evaluation functors}  
that evaluate each graph $G$, considered as a functor $G:\C_{\gr}\to\Set$, 
at the objects $v$ and $e$ of $\C_{\gr}$, respectively: 
$V(G)=G(v)$ and $E(G)=G(e)$. 
It follows that the functors $V$, $E$ and $S=V+E$ preserve 
limits and colimits, in particular POs and PBs. 

% A
Since $\Mod(\Sp)$ is the category of models of an equational 
specification $\Sp$, it has limits and colimits.
% T
Let $S'$ be some set of sorts of $\Sp$,
i.e., some subset of $S$. 
Let $TA=\sum_{s\in S'} A_s$ be the disjoint union of the 
carriers sets $A_s$ for $s$ in $S'$. % morphisms? ***
The functor $T$ is such that $TA=\sum_{s\in S'} A_s$ 
for some set of sorts $S'$ of $\Sp$. 
It preserves limits (but it does not preserve colimits, 
in general), thus it preserves PBs. 

% U_{\G}
Let $A_{\emptyset}$ be the empty model of $\Sp$,
defined by the fact that all its carriers are empty:
$A_{\emptyset,s}=\emptyset$ for each $s\in S$.
Then $A_{\emptyset}$ is an initial object of the category $\Mod(\Sp)$ 
and $TA_{\emptyset}=\emptyset$. 
The functor $U_{\G}$ has a left adjoint, which maps each 
$G$ to $(G,A_{\emptyset},\alpha_{\emptyset})$ where 
$\D(\alpha_{\emptyset})=\emptyset$. % morphisms? ***
% proof? ***
Thus, the functor $U_{\G}$ preserves limits, especially PBs. 

% $U_{\A}
Let $G_{\emptyset}$ be the empty graph,
it is an initial object of the category $\Gr$ 
and $S(G_{\emptyset})=\emptyset$.  
The functor $U_{\A}$ has a left adjoint, which maps each 
$A$ to $(G_{\emptyset},A,\alpha_{\emptyset})$ where 
$\D(\alpha_{\emptyset})=\emptyset$. % morphisms? ***
% proof? *** 
Thus, the functor $U_{\A}$ preserves limits, especially PBs. 
\end{proof}

\end{document}